\documentclass[prl,twocolumn, aps,showpacs,superscriptaddress]{revtex4}

\usepackage{latexsym}
\usepackage{amsmath}
\usepackage{amssymb}
\usepackage{amsfonts}

\usepackage{revsymb}

\usepackage{graphicx}
\usepackage[caption=false]{subfig}

\usepackage{graphicx}
\usepackage{tikz}
\usetikzlibrary{decorations.markings,decorations.pathreplacing,calc,fadings, shadings, patterns}

\usepackage{amsthm,amsfonts,amsmath, amscd}

                                \usepackage{verbatim}

\newcommand{\cL}{{\mathcal L}}

\newtheorem{proposition}{Proposition}
\newtheorem{theorem}{Theorem}

\newtheorem{question}{Question}

\newcommand{\ket}[1]{\left\vert #1\right\rangle}
\newcommand{\bra}[1]{\left\langle #1\right\vert}

\newcommand{\Tr}{\mbox{Tr}}

\newcommand{\be}{\begin{equation}}
\newcommand{\ee}{\end{equation}}
\newcommand{\bea}{\begin{eqnarray}}
\newcommand{\eea}{\end{eqnarray}}
\newcommand{\beann}{\begin{eqnarray*}}
\newcommand{\eeann}{\end{eqnarray*}}

\usepackage{color}

\begin{document}
\title{Entanglement rates for bipartite open systems}
\author{Anna Vershynina\\
\small{Institute for Quantum Information, RWTH Aachen University, Aachen, Germany}\\
\small{annavershynina@gmail.com}}

\date{\today}

\begin{abstract}
We provide upper bound on the maximal rate at which irreversible quantum dynamics can generate entanglement in a bipartite system. The generator of irreversible dynamics consists of a Hamiltonian and dissipative terms in Lindblad form. The relative entropy of entanglement is chosen as a measure of entanglement in an ancilla-free system. We provide an upper bound on the entangling rate which has a lnarithmic dependence on a dimension of a smaller system in a bipartite cut. We also investigate the rate of change of quantum mutual information in an ancilla-assisted system and provide an upper bound independent of dimension of ancillas.
\end{abstract}


\maketitle

\section{Introduction}

The problem addressed in this paper is, given some Hamiltonian and dissipative interactions between two or more subsystems, what is the maximal rate at which an entanglement can be generated in time. The problem of upper bounding the entangling rate is called \textit{small incremental entangling}. See the next section for precise definition.

The same question for closed bipartite system evolving under a unitary dynamics was answered by Bravyi in \cite{Bravyi} and by Acoleyen et al. in \cite{Acoleyen} and Audenaert \cite{Audenaert}  for the general case in the presence of ancillas. Let us say that two parties, Alice and Bob, have access to systems $A$ and $B$ respectively together with ancilla systems $a$ and $b$ respectively. The system starts in a pure state $\ket{\Psi}_{aABb}$ and evolves according to a Hamiltonian $H_{AB}$ that acts only on systems $A$ and $B$. Since the state of a system $\ket{\Psi(t)}_{aABb}=\exp\{itH_{AB}\}\ket{\Psi}_{aABb}$ stays pure, one may calculate the entanglement entropy as a measure of entanglement in a bipartite cut between Alice's $aA$ and Bob's systems $Bb$, $E(t)=-\Tr(\rho_{aA}(t)\ln\rho_{aA}(t))$, where $\rho_{aA}(t)=\Tr_{Bb}\ket{\Psi(t)}\bra{\Psi(t)}_{aABb}$. The \textit{entangling rate} is a time derivative of the entanglement entropy at time $t=0$, $ \Gamma(\Psi, H)={dE(\rho(t))}/{dt}|_{t=0}.$ The small incremental entangling problem aims at finding an upper bound on the maximal entangling rate $\Gamma(H)=\sup_{\ket{\Psi}_{aABb}}\Gamma(\Psi, H)$ that is independent of dimensions of ancillas and the initial state $\ket{\Psi}_{aABb}$.  The problem of maximizing the entangling rate of a bipartite system in the presence of ancillas evolving under unitary dynamics was studied by many authors.

In a case when $A$ and $B$ are qubits, Childs et al \cite{CLV06} gave upper bounds for 
an entangling rate and showed that they are independent of ancillas $a$ and $b$. Wang and Sanders \cite{WS03} considered systems $A$, $B$ and ancillas of arbitrary dimensions and proved that the entangling rate is upper bounded by
$\Gamma(H)\leq\beta\approx 1.9123$, for a self-inverse product Hamiltonian.  Bravyi \cite{Bravyi} proved that in a general case with no ancillas the entangling rate is bounded by $\Gamma(H)\leq c(d)\|H\|\ln d,$ where $d=\min(d_A, d_B)$ is the smallest dimension of the interacting subsystems and $c$ is a constant close to 1. For an arbitrary bipartite Hamiltonian Bennett et al \cite{Bennett} proved that the upper bound on the ancilla-assisted entanglement is independent of the ancilla dimensions
$\Gamma(H)\leq c d^4\|H\|$, where $c$ does not depend on $a$ or $b$. The bound was improved by Lieb and Vershynina \cite{Lieb} providing an upper bound $\Gamma(H)\leq 4\|H\| d$ for an arbitrary Hamiltonian in ancilla-assisted system.  Finally the question was answered by Acoleyen et al. in \cite{Acoleyen} arriving at  $\Gamma(H)\leq18\|H\|\ln d.$ Few months later an independent proof was presented by Audenaert \cite{Audenaert} that gives an upper bound $\Gamma(H)\leq 8\|H\|\ln d.$ The bound with lnarithmic dependence is optimal, since one can find a particular Hamiltonian $H_{AB}$ for which they there is an equality. See Marien et al. \cite{Marien} for a detailed review on the entangling rates for bipartite closed systems.

In this paper we go further by considering open systems evolving under irreversible dynamics and aim to provide an upper bound on the entangling rate of a bipartite system.  The behavior of entanglement under non-unitary evolution was first studied by Rajagopal and Rendell \cite{Raja} and  Zyczkowski et al \cite{Zycz}. Rajagopal and Rendell studied a dynamic evolution of a  pair of initially entangled harmonic oscillators in the presence of local environments.  Zyczkowski et al considered a bipartite system consisting of two spin-1/2 particles, one of which is subjected to periodic actions of a quantum channel representing the interaction with environment, and the whole system undergoes a sequence of global unitary evolutions. Both papers found that the entanglement may exhibit revivals in time and that it might vanish at finite times. For further research on revival of entanglement and its collapse see, for example, \cite{Ficek06}-\cite{Lopez}. Another motivation to study entanglement rate problem in open system arises from investigation of production and robustness of entanglement in open systems, see, for example,  \cite{Bandy}, \cite{Marzolino}. See \cite{Aolita} for a comprehensive review of the study of dynamics of entanglement in open systems.

To investigate the entangling rate in an open system we can not consider an entropy of entanglement as an entanglement measure, since for a system starting in a pure state and evolving under non-unitary evolution a time-dependent state may no longer stay pure. 

To quantify the entanglement in ancilla-free system we take the relative entropy of entanglement, which reduces to the entropy of entanglement on pure states. We derive an upper bound on the entangling rate that has a lnarithmic dependence on the dimension of a smallest system, as it was in the case of a unitary dynamics. These result rises a question of finding a class of entanglement measures for which the same holds true together with another question whether the same bound can be derived in the ancilla-assisted case. 

In ancilla-assisted system we consider quantum mutual informal and provide an upper bound on its derivative that is independent of ancillas.

The paper is organized as follows: in Section 2 we discuss a state evolution in an open system, talk about entanglement measures, describe the small incremental entangling problem in a closed system and conjecture the problem for the open system. In Section 3 we provide an upper bound on the entangling rate for a relative entropy of entanglement in ancilla-free system that has lnarithmic dependence on the smaller system. In Section 4 we derive an upper bound on the derivative of quantum mutual information in ancilla-assisted system.

\section{Preliminaries}

Suppose that two parties, say Alice and Bob, have control over systems $A$ and $B$. In ancilla-assisted entangling both parties have access to additional subsystems, called local ancillas, i.e. Alice is in control of two systems $A$ and $a$ and Bob is in control of $B$ and $b$. 

Alice and Bob start with a pure state $\rho_{aABb}(0)=\ket{\Psi}\bra{\Psi}_{aABb}$.
A time dependent joint state of Alice and Bob $\rho(t)$ satisfies the following differential equation  
\begin{equation}\label{irrev}
\frac{d}{dt}\rho(t)=\cL_{AB}(\rho(t)).
\end{equation}
The generator of the dynamics $\cL_{AB}$ acts only on systems $A$ and $B$, and consists of both Hamiltonian part $H_{AB}$ and a dissipative part in Lindblad form with operators $L_{AB}$ and takes the following form \cite{Lindblad}
\begin{align}\label{super}
\cL_{AB}(\rho)=&-i[H_{AB},\rho]+\sum_\alpha\Bigl( L^{AB}_\alpha\rho (L^{AB}_\alpha)^\dagger\\
&-\frac{1}{2} (L^{AB}_\alpha)^\dagger L^{AB}_\alpha\rho-\frac{1}{2}\rho  (L^{AB}_\alpha)^\dagger L^{AB}_\alpha\Bigr).\nonumber
\end{align}
Here $H_{AB}$, $\{L_\alpha^{AB}\}_\alpha$ are bounded linear operators and $H_{AB}$ is also Hermitian. We will be dropping the (sub) superscription notation $AB$ throughout the paper for easier notation when it is clear from the context which systems an operator acts on.

We are going to consider entanglement measure $E(\cdot)$ that satisfy the following assumptions:
\begin{enumerate}
\item $E$ vanishes on product states
\item $E$ is invariant under local unitary operations
\item $E$ can not increase under LOCC operations
\item for any $\rho$ we have $E(\rho\otimes \ket{\Phi_d}\bra{\Phi_d})=E(\rho)+E(\ket{\Phi_d})$, where $\ket{\Phi_d}=\frac{1}{\sqrt{d}}\sum_{j=1}^d\ket{j}_{1}\ket{j}_{2}$ is the maximally entangled state shared between two parties.
\end{enumerate}

There are several entanglement measures that satisfy these conditions. In this paper we will focus on entanglement measures that  reduce to the entropy of entanglement on pure states\\

$\text{5. for }\rho=\ket{\phi}\bra{\phi}_{12}, \text{ we have } E(\rho)=-\Tr_1(\rho_1\ln\rho_1)=-\Tr_2(\rho_2\ln\rho_2).$\\

In case of a unitary dynamics, i.e. $L_\alpha=0$ for all $\alpha$, the total change of entanglement though time is no more than $2\ln d$ with $d=\min\{d_A, d_B\}$, see \cite{Bennett}, \cite{Marien}. Property 4 of the entanglement measure plays a significant role in the bound of the total change of entanglement. In the case of irreversible dynamics the total change of entanglement depends on the measure of entanglement and it is unknown for many measures. For example, a relative entropy of entanglement (see definition in (\ref{rel_entropy})) gives an trivial bound of $4\ln D$ with $D=\min\{d_{aA}, d_{Bb}\}$, and negativity $\mathcal{N}$ gives a linear bound in $d_{aA}$ or $d_{Bb}$, see Tiersch et al. \cite{Tiersch}.  

A unitary case gave motivation to derive a bound on an instantaneous rate of change of the entanglement. This problem is formulated using the entangling rate. 
The case of a unitary dynamics \textit{{entangling rate}} is defined as
$$ \Gamma(\Psi, H)=\frac{dE(\rho(t))}{dt}\bigg|_{t=0}.$$
For a system evolving according to a unitary evolution the following theorem was conjectured in \cite{Bravyi} and was first proved in \cite{Acoleyen}.

\begin{theorem}\label{closed} \cite{Acoleyen} \textbf{Small incremental entangling for unitary dynamics.}
Denote $d=\min\{d_A, d_B\}$. For a system evolving according to a Hamiltonian $H_{AB}$ there exists a constant $c$ such that the entangling rate is bounded above by
$$\Gamma(\Psi, H)\leq c \|H\|\ln d,$$
where $c$ is independent of the dimensions of systems $A$ and $B$, ancillas $a$, $b$, the Hamiltonian $H_{AB}$ and the initial state $\ket{\Psi}$.
\end{theorem}

The upper bound on the maximal entangling rate was investigated in many papers \cite{Bravyi}, \cite{CLV06}-\cite{Lieb} and was first proved in \cite{Acoleyen}  with a different proof provided in \cite{Audenaert}, see also \cite{Marien} for detailed review of this case.

We ask the same question for a system evolving under an irreversible dynamics (\ref{irrev}). In general case an entanglement measure may not be differentiable. Define the \textbf{entangling rate} for time interval $\Delta t >0$ as
$$\Gamma(\Psi, \cL, \Delta t):=\frac{E(\rho(t))-E(\rho(0))}{\Delta t}.$$
In case when the entanglement measure is differentiable, the entangling rate is defined as
$$ \Gamma(\Psi, \cL)=\frac{dE(\rho(t))}{dt}\bigg|_{t=0}.$$
 Similar to a system evolved under unitary evolution we are looking for the upper bound on entangling rate that is independent of ancillas and has a proper dependence on dimension of systems
 
\begin{question} \textbf{Small incremental entangling.}
Denote $d=\min\{d_A, d_B\}$. For which entanglement measures there exists a constant $c$ and a non-negative non-decreasing function $f(\cdot)$ such that for any $\epsilon>0$ there exists $\delta>0$ such that for any $\Delta t<\delta$  the entangling rate is bounded above by
\begin{equation}\label{con}
\Gamma(\Psi, \cL, \Delta t)\leq c \|\cL\| f(d)+\epsilon,
\end{equation}
where $c$ is independent of the dimensions of systems $A$, $B$, ancillas $a$, $b$, the generator $\cL$  and the initial state $\ket{\Psi}_{aABb}$.
\end{question}

Note that the intuitive argument to consider the open system dynamics on system $AB$ as a part of the unitary dynamics on the bigger system $ABE$ doesn't work in this case, because there is no bound of the type $\|\cL_{AB}\|\geq C\|F_{ABE}\|$ with $C>0$ being independent of the generator $\cL_{AB}$ and the induced Hamiltonian $F_{ABE}$. The bound of this type does not exist because in a case when $\cL_{AB}\equiv 0$ ($H= I$, $L_\alpha\equiv 0$) and $F_{ABE}= I_{ABE}$ constant $C$ must be zero.

We answer Small Incremental Entangling question for the relative entropy of entanglement in an ancilla-free case and for the quantum mutual information in ancilla-assisted system.

\section{Relative entropy of entanglement in ancilla free system}

Assume throughout this section that there are no ancillas, $d_a=d_b=1$, and that $d:=d_B\leq d_A$. For states and operators that act on both systems $A$ and $B$ we will drop the subscript notation $AB$.

The relative entropy of entanglement \cite{Vedral} of a state $\rho(t)$ is given by
\begin{align}\label{rel_entropy}
D(\rho(t)):&=\min_{\sigma sep}D(\rho(t)||\sigma)\\
&=\min_{\sigma sep}\Tr\Bigl(\rho(t)\ln \rho(t)-\rho(t)\ln\sigma\Bigr),\nonumber
\end{align}
where $\sigma=\sum_j\alpha_j\sigma_{A}(j)\otimes\sigma_{B}(j)$ with $\sum_j\alpha_j=1$ is a separable state between Alice and Bob's systems.
For pure states $\rho:=\rho(0)=\ket{\Psi}\bra{\Psi}$ the relative entropy of entanglement is an entropy of entanglement of this state \cite{Vedral}. 

In this section we prove the following theorem.

\begin{theorem}\label{relative} For any $\epsilon>0$ there exists $\delta>0$ such that for any $\Delta t<\delta$ the entangling rate for the relative entropy of entanglement has the following upper bound
$$\Gamma_R(\Psi, \cL,  \Delta t) \leq 4\Bigl(\|H\|+86\sum_\alpha\|L_\alpha\|^2\Bigr)\ln d+\epsilon,$$
where $d=\min(d_A, d_B)$.
\end{theorem}

First, let us look closely on the relative entropy of entanglement for pure states. Let state $\ket{\Psi}$ have Schmidt decomposition 
\begin{equation}\label{rho_in}
\ket{\Psi}=\sum_{n=1}^d\sqrt{p_n}\ket{\phi_n}_{A}\ket{\psi_n}_{B}.
\end{equation} 
Then the relative entropy of entanglement is achieved by a state \cite{Vedral} 
\begin{equation}\label{sigma}
\sigma_0=\sum_{n=1}^d p_n\ket{\phi_n}\bra{\phi_n}\otimes\ket{\psi_n}\bra{\psi_n}.
\end{equation}

\begin{proposition}\label{prop::convex}

For states $\rho=\ket{\Psi}\bra{\Psi}$ and $\sigma_0$ defined in (\ref{rho_in}) and (\ref{sigma}) respectively, there exists a mixed state $\mu$ such that
$$\sigma_0=\frac{1}{d}\rho+(1-\frac{1}{d})\mu. $$
\end{proposition}

\begin{proof}
The statement of the proposition is equivalent to inequality
\begin{equation}\label{positive}
d\sigma_0\geq \rho=\ket{\Psi}\bra{\Psi},
\end{equation}
which in turn is equivalent to inequality
$$Z:=d\sum_np_n\ket{\phi_n\psi_n}\bra{\phi_n\psi_n}-\sum_{n,k}\sqrt{p_np_k}\ket{\phi_n\psi_n}\bra{\phi_k\psi_k}\geq 0.$$
The last inequality is equivalent to a statement that for any state, which can be written as $\ket{\Omega}=\sum_{nk}\sqrt{\alpha_{nk}}\ket{\phi_n}\ket{\psi_k}$, the expectation value of the observable $Z$ is positive, i.e.
\begin{equation}\label{pos}
\bra{\Omega}Z\ket{\Omega}=d\sum_n p_n\alpha_{nn}-\Bigl(\sum_n \sqrt{p_n\alpha_{nn}}\Bigr)^2\geq 0.
\end{equation}

Since a function $f(x)=x^2$ is convex, we have that for any $x_n, n=1,...,d$, $$\Bigl(\sum_{n=1}^d\frac{1}{d}x_n\Bigr)^2\leq \sum_n \frac{1}{d}x_n^2.$$
Taking $x_n=\sqrt{p_n\alpha_{nn}}$ in the above inequality we arrive at inequality (\ref{pos}), which in turn gives (\ref{positive}), proving the proposition.

\end{proof}

Now we are ready to prove Theorem \ref{relative}.
\begin{proof}
The relative entropy of entanglement of a state $\rho(t=0)$ at time zero is the relative entropy of that state $D(\rho)=D(\rho||\sigma_0)=E(\Psi)$, where $\sigma_0$ is defined in (\ref{sigma}).  For any time $t$ we have that $D(\rho(t))\leq D(\rho(t)||\sigma_0)$. Therefore for any $\epsilon>0$ there exists $\delta>0$ such that for any $\Delta t<\delta$
\begin{equation}\label{entangling}
\Gamma_R(\Psi, \cL, \Delta t)\leq \frac{d}{dt}D(\rho(t)||\sigma_0)\bigg|_{t=0}+\epsilon.
\end{equation}

The derivative of the relative entropy $D(\rho(t)||\sigma_0)$ can be calculated as follows

\begin{align}
&\frac{d}{dt}D(\rho(t)||\sigma_0)\bigg|_{t=0}\nonumber\\
&=\, \Tr(\dot{\rho}(t)|_{t=0}\ln\rho-\dot{\rho}(t)|_{t=0}\ln\sigma_0)\nonumber\\
&=-i\Tr([H,\rho]\ln\rho)+i\Tr([H,\rho]\ln\sigma_0)\nonumber\\
&+\sum_\alpha \Tr\Bigl(L_\alpha\rho L^\dagger_\alpha\ln \rho-\frac{1}{2}\{L^\dagger_\alpha L_\alpha, \rho\}\ln \rho\Bigr)\nonumber\\
&-\sum_\alpha \Tr\Bigl(L_\alpha\rho L^\dagger_\alpha\ln\sigma_0\Bigr)
+\sum_\alpha\Tr\Bigl(\frac{1}{2}\{L^\dagger_\alpha L_\alpha, \rho\}\ln\sigma_0\Bigr)\nonumber\\
&= \frac{1}{p}i\Tr\Bigl(H[p\rho, \ln\sigma_0]\Bigr)\label{H_term}\\
&-\frac{1}{2p}\sum_\alpha \Tr\Bigl(L_\alpha^\dagger[L_\alpha(p\rho),\ln\sigma_0]\Bigr)\label{L_term}\\
&+\frac{1}{2p}\sum_\alpha\Tr\Bigl(L_\alpha[(p\rho) L_\alpha^\dagger,\ln\sigma_0]\Bigr)\nonumber\\
&-\sum_\alpha \Tr\Bigl(L_\alpha^\dagger[L_\alpha\rho, \ln\rho]\Bigr).\nonumber
\end{align}
Here for the first equality we used that $\Tr(\dot{\rho}(t)|_{t=0})=\Tr(\cL(\rho))=0$, which can be seen from the expression of the Lindbladian generator (\ref{super}) and the cyclicity of a trace. For the third equality we denoted $p:=1/d$.
Note that the third term in the last equality is a complex conjugate of the second one 
$$\Tr(L^\dagger[LX, \ln Y])=\overline{\Tr(L[X L^\dagger, \ln Y])},$$
where $X:=p\rho$, $Y:=\sigma_0=p\rho+(1-p)\mu$ according to Proposition \ref{prop::convex}. 
The fourth term can be brought in the following form, for any $0<p$ we have
\begin{align*}
\Tr(L^\dagger[L\rho, \ln \rho]) &=\frac{1}{p}\Tr(L^\dagger[L(p\rho), \ln \rho])\\
&=:\frac{1}{p}\Tr(L^\dagger[LX, \ln Y]),
\end{align*}
where $0\leq X\leq Y$, $\Tr X=p$, $\Tr Y=1$. Therefore
\begin{align*}
\frac{d}{dt}D(\rho(t)||\sigma_0)\bigg|_{t=0}\leq& \frac{1}{p}|\Tr(H[X,\ln Y])|\\
&+2\frac{1}{p}\sum_\alpha |\Tr(L_\alpha^\dagger[L_\alpha X, \ln Y])|,
 \end{align*}
where $\Tr X=p$, $\Tr Y=1$ and $0\leq X\leq Y.$

Consider two terms (\ref{H_term}) and (\ref{L_term}) separately. 

1. Term (\ref{H_term}) is linear in $H$, so we may pull out the norm $\|H\|$, then the term becomes exactly the one that arises in a unitary case
\begin{equation}\label{unitary}
\Tr(P [X, \ln Y])
\end{equation}
for $\|P\|=1$, $\Tr X=p$, $\Tr Y=1$ and $0\leq X\leq Y.$ Acording to \cite{Audenaert} the following upper bound holds
\begin{equation*}
\Tr(H [p\rho, \ln \sigma_0])\leq -2\Bigl(p\ln p+(1-p)\ln(1-p)\Bigr)\|H\|.
\end{equation*}
For any $p\leq 1/2$ one has $-\Bigl(p\ln p+(1-p)\ln(1-p)\Bigr)\leq -2p\ln p$ and therefore
\begin{equation}\label{unitary_en}
\Tr(H [\rho, \ln \sigma_0])\leq -4\ln p\|H\|=4\ln d \|H\|.
\end{equation}

2. Pulling out the norms of $L_\alpha$ term (\ref{L_term}) becomes
$$\Tr(L^\dagger[LX, \ln Y])=\|L\|^2\, \Tr(\tilde{L}^\dagger[\tilde{L}X, \ln Y]), $$
with $\|\tilde{L}\|=1$, $\Tr X=p$, $\Tr Y=1$ and $0\leq X\leq Y$.

Any operator can be written as a sum of hermitian and anti-hermitian parts, let $\tilde{L}^\dagger=2{P}-I+i(2{R}-I)$ with $0\leq {P},{R}\leq I$. Denote $Q:=\tilde{L}X$, then 
\begin{align*}
|\Tr(\tilde{L}^\dagger[Q, \ln Y])|&\leq |2\Tr(P[Q, \ln Y])+2i\Tr(R[Q, \ln Y])|\\
&\leq 4|\Tr(\tilde{P} [Q, \ln Y])|,
\end{align*}
where $0\leq \tilde{P}\leq I$ denotes either $P$ or $R$. 

Expression $|\Tr(P [Q,\ln Y])|$ is similar to the expression of the entangling rate for a unitary dynamics (\ref{unitary}), which allows us to follow the argument presented in \cite{Acoleyen} to obtain an upper bound on this type of expressions. For completeness sake we present a full analysis here, which globally resembles an argument in \cite{Acoleyen}, but since $Q=LX$ is not even Hermitian some details need to be different in order to derive an upper bound. In the future we will drop tilde above $\tilde{L}$ in expression for $Q$.

Let $y_j$ denote the eigenvalues of $Y$ in decreasing order $1\geq y_1\geq y_2\geq ...\geq y_N\geq 0$ with corresponding eigenstate $\ket{\phi_j}$. Denote the matrix elements of $Q$ in this eigenbasis as $Q_{ij}=\bra{\phi_i}Q\ket{\phi_j}$, similarly define the matrix elements of $P$ as $P_{ij}=\bra{\phi_i}P\ket{\phi_j}$. Then
\begin{equation}\label{norm}
|\Tr(P [Q,\ln Y])|=\Big|\sum_{i<j}\ln\frac{y_i}{y_j}(Q_{ij}P_{ji}-Q_{ji}P_{ij})\Big|.
\end{equation}
Group the eigenvalues of $Y$ in the following intervals
\begin{align}
1>y_{i_1}\geq p&\ \ \ \ \ 1\leq i_1\leq n_1\nonumber\\
p>y_{i_2}\geq p^2&\ \ \ \ \ n_1<y_{i_2}\leq n_2\nonumber\\
...&\label{intervals}\\
p^{k-1}>y_{i_k}\geq p^k&\ \ \ \ \ n_k<i_k\leq N.\nonumber
\end{align}
Some of the intervals can be empty, in this case then $n_{k-1}=n_k$. Now rearrange the sum (\ref{norm}) in the following way, for $\lambda_j\in(n_{k-1}, n_k)$
\begin{align}
\sum_{i<j}&= \sum_{1\leq i_1<j_1\leq n_2} + \sum_{n_1\leq i_2<j_2\leq n_3}+...+\sum_{n_{k-1}\leq i_k<j_k\leq N}\label{sum}\\
&-\sum_{n_1\leq i_2<j_2\leq n_2}-\sum_{n_2\leq i_3<j_3\leq n_3}-...-\sum_{n_{k-1}\leq i_k<j_k\leq n_k}\label{sum_2}\\
&+\Big(\sum_{\lambda_1, \lambda_{k>2}}+\sum_{\lambda_2, \lambda_{k>3}}+...+\sum_{\lambda_{k-2}, \lambda_{k}} \Bigr).\label{sumlast}
\end{align}
The first two lines contain sums that run over the one or two consequent intervals in (\ref{intervals}). The sum on the last line runs over pairs that are separated by at least one interval in (\ref{intervals}).

The sums in the first line (\ref{sum}) can be calculated in the following way. Let us bound, for example, the first sum
\begin{align}
&|\sum_{i=1}^{n_2}\sum_{j=i+1}^{n_2}\ln\frac{y_i}{y_j}(Q_{ij}P_{ji}-Q_{ji}P_{ij})  |=|\Tr(\tilde{P}[\tilde{Q}, \ln\tilde{Y}])|\nonumber\\
&\leq 2|\Tr(\tilde{P}[\tilde{X}, \ln\frac{\tilde{Y}}{\tilde{y}_{min}}])|+4|\Tr(\tilde{P}[\tilde{T}\tilde{X}, \ln\frac{\tilde{Y}}{\tilde{y}_{min}}])|.\label{partial_sum}
\end{align}
Here we wrote out  operator $L=2P-I+i(2R-I)$ as a sum of hermitian and anti-hermitian parts and denoted a Hermitian operator $T$ as either operator $P$ or $R$. We also added a zero term by scaling $\tilde{Y}$. Tilde above an operator denote a restriction of this operator on a space spanned by eigenstates $\ket{\phi_j}$ for $j=1,..., n_2$.

The first term in (\ref{partial_sum}) is bounded by
\begin{align*}
2|\Tr(\tilde{P}[\tilde{X}, \ln\frac{\tilde{Y}}{\tilde{y}_{min}}])|&\leq \|[\tilde{X}, \ln\frac{\tilde{Y}}{\tilde{y}_{min}}]\|_1\\
&\leq  \ln\frac{{\tilde{y}_{max}}}{\tilde{y}_{min}}\|\tilde{X}\|_1\\
&\leq 2(p_1+p_2)\ln(1/p).
\end{align*}
Here for the first inequality we used the fact that for any Hermitian operator $A$
\begin{equation}\label{trace_norm}
\|A\|_1=\max_{\|P\|\leq 1}|\Tr(PA)|=2\max_{0\leq P\leq I}|\Tr(PA)|.
\end{equation} For the second inequality we used Kittaneh inequality \cite{Kittaneh} for a positive compact operator $A$ and a unitarily invariant norm $|||[A,X]|||\leq \|A\|\ |||X|||.$ On the last inequality we used $p^2<y_i/y_j<1/p^2$ and $\|X\|_1=\Tr\tilde{X}=p_1+p_2$, where $\Tr X=\sum_k p_k =p.$

The second term in (\ref{partial_sum}) can be bounded similarly
\begin{align*}
&4|\Tr(\tilde{P}[\tilde{T}\tilde{X}, \ln\frac{\tilde{Y}}{\tilde{y}_{min}}])|\\
&\leq 4|\Tr(\tilde{P}\tilde{T}\tilde{X} \ln\frac{\tilde{Y}}{\tilde{y}_{min}})|+4|\Tr(\tilde{T}\tilde{X}\tilde{P} \ln\frac{\tilde{Y}}{\tilde{y}_{min}})|\\
&=4\|\ln\frac{{\tilde{Y}}}{\tilde{y}_{min}} \|\Bigl( |\Tr(W\tilde{P}\tilde{T}\tilde{X})|+|\Tr(\tilde{P} W\tilde{T}\tilde{X})|\Bigr)\\
&\leq 8\ln\frac{{\tilde{y}_{max}}}{\tilde{y}_{min}} \|\tilde{X} \|_1\leq 16(p_1+p_2)\ln(1/p).
\end{align*}
Here for the equality we denoted $W:=\ln\frac{\tilde{Y}}{\tilde{y}_{min}}/\|\ln\frac{\tilde{Y}}{\tilde{y}_{min}}\|$. Since $\|W\|=1$, we used equality (\ref{trace_norm}) for the second inequality, since $\|W\tilde{P}\tilde{T}\|\leq 1 $ and $\| \tilde{P} W\tilde{T}\|\leq 1$.

Therefore line (\ref{sum}) can be bounded above by
\begin{equation}\label{first}
18\Bigl(p_1+p_k+2\sum_{j=2}^{k-1}p_j\Bigr)\ln(1/p)\leq 36\, p\ln(1/p).
\end{equation}

Line (\ref{sum_2}) can be bounded in a similar fashion. The first term, for example, gets a bound
\begin{align*}
|\sum_{i=n_1}^{n_2}\sum_{j=i+1}^{n_2}\ln\frac{y_i}{y_j}(Q_{ij}P_{ji}-Q_{ji}P_{ij})  |&\leq 5\, \Tr(\tilde{X})\ln\frac{\tilde{y}_{max}}{\tilde{y}_{min}}\\
&\leq 5 p_2\ln(1/p).
\end{align*}
Therefore the whole second line (\ref{sum_2}) can be bounded by
\begin{equation}\label{second}
5(p_2+...+p_k)\ln(1/p)\leq 5p \ln(1/p).
\end{equation}

To bound the last line (\ref{sumlast}) define $Z=Y^{-1/2}QY^{-1/2}$. Then the last line (\ref{sumlast}) reads
\begin{align}
&|\tilde{\sum}_{i<j}\ln\frac{y_i}{y_j}y_i^{1/2}y_j^{1/2}(Z_{ij}P_{ji}-Z_{ji}P_{ij})|\nonumber\\
&\leq  \Bigl(\tilde{\sum}_{i<j} \ln\frac{y_i}{y_j}y_i^{1/2}y_j^{1/2} Z_{ij}\overline{Z_{ij}}\Bigr)^{1/2}\nonumber\\
&\times\Bigl(\tilde{\sum}_{i<j} \ln\frac{y_i}{y_j}y_i^{1/2}y_j^{1/2} |P_{ij}|^2\Bigr)^{1/2}\nonumber\\
&+\Bigl(\tilde{\sum}_{i<j} \ln\frac{y_i}{y_j}y_i^{1/2}y_j^{1/2} Z_{ji}\overline{Z_{ji}}\Bigr)^{1/2}\nonumber\\
&\times\Bigl(\tilde{\sum}_{i<j} \ln\frac{y_i}{y_j}y_i^{1/2}y_j^{1/2} |P_{ji}|^2\Bigr)^{1/2}\nonumber\\
&\leq  \ p^{1/2}\ln(1/p) \Bigl(\sum_{i,j=1}^N y_i Z_{ij} (Z^\dagger)_{ji} \Bigr)^{1/2}\Bigl(\sum_{i,j=1}^N y_i |P_{ij}|^2\Bigr)^{1/2}\label{ineq}\\
&+ p^{1/2}\ln(1/p) \Bigl(\sum_{i,j=1}^N y_i Z_{ji} (Z^\dagger)_{ij} \Bigr)^{1/2}\Bigl(\sum_{i,j=1}^N y_i |P_{ji}|^2\Bigr)^{1/2}\nonumber
\end{align}
Here for the first inequality we used Cauchy-Schwarz inequality. For the second inequality  we used that $x^{1/2}\ln(1/x)\leq p^{1/2}\ln(1/p)$ if $x\leq p\leq 1/e^2$, with $x=y_j/y_i$. 

Since $0\leq P\leq I$, we have that $0\leq P^2\leq I$ and therefore for every $i$ we have $1\geq(P^2)_{ii}=\sum_j P_{ij}P_{ji}=\sum_j|P_{ij}|^2$. Using $\sum_i y_i=\Tr Y=1$ we continue calculations 
\begin{align}\label{cont}
(\ref{ineq})&\leq p^{1/2}\ln(1/p)\Bigl(\Tr(YZ Z^\dagger)^{1/2}+\Tr(YZ^\dagger Z)^{1/2} \Bigr).
\end{align}
Considering each trace separately, recall that $Z=Y^{-1/2}LXY^{-1/2}$. Then $$\Tr(YZ Z^\dagger)=\Tr(L^\dagger LX Y^{-1}X).$$ Denote $R:=XY^{-1}X$. Since $0\leq L^\dagger L\leq I$, $0\leq X^{1/2}Y^{-1}X^{1/2}\leq I$, we have that $0\leq R\leq X$ and $R^{1/2}(I-L^\dagger L)R^{1/2}\geq 0$. Therefore $$\Tr(YZ Z^\dagger)=\Tr(L^\dagger LR)\leq \Tr(R)\leq \Tr(X)=p.$$
The second term in (\ref{cont}) can be estimated in the following way
$$\Tr(YZ^\dagger Z)=\Tr(XL^\dagger Y^{-1}LX)=:\Tr(A^\dagger A),$$
where $A:=Y^{-1/2}LX$. Therefore
\begin{align*}
\Tr(YZ^\dagger Z)&=\sum_j |\lambda_j(A)|^2=\sum_j |\lambda_j(X^{1/2}AX^{-1/2})|^2\\
&=\sum_j |\lambda_j(X^{1/2}Y^{-1/2}LX^{1/2})|^2=:\Tr(B^\dagger B),
\end{align*}
where $\lambda_j(A)$ is $j$-th eigenvalue of  operator $A$ and $B:=X^{1/2}Y^{-1/2}LX^{1/2}.$ Here we used the fact that $A$ and $P^{1/2}AP^{-1/2}$ have same eigenvalues for any $P$ since they have the same characteristic polynomial. Therefore writing out $B^\dagger B$ we obtain
\begin{align*}
 \Tr(YZ^\dagger Z)&=\Tr(Y^{-1/2}X Y^{-1/2} LXL^\dagger)\leq \Tr(LXL^\dagger)\\
 &\leq \Tr(X)=p,
 \end{align*}
here we used that $LXL^\dagger=(X^{1/2}L^\dagger)^\dagger(X^{1/2}L^\dagger)\geq 0$.
Thus line (\ref{sumlast}) is bounded above by
\begin{equation}\label{third}
(\ref{sumlast})\leq 2 p\ln(1/p).
\end{equation}

Combining bounds (\ref{first})-(\ref{third}) all lines  (\ref{sum})-(\ref{sumlast}) in  sum (\ref{norm}) are bounded above by the following expression
\begin{equation}\label{L-term}
|\Tr(L^\dagger[LX, \ln Y])|\leq 172\|L\|^2 p\ln(1/p),
\end{equation}
where $p=1/d$.
Using this bound together with (\ref{unitary_en}), entangling rate (\ref{entangling}) is therefore bounded above by 
\begin{align*}
\Gamma_R(\Psi, \cL, \Delta t)&\leq 4\Bigl(\|H\|+86\sum_\alpha\|L_\alpha\|^2 \Bigr)\ln d+\epsilon,
\end{align*}
where $d=\min(d_A, d_B)$.
\end{proof}

Note that in the ancilla-free system if all Lindblad operators $L_\alpha$ vanish, i.e. in the case of unitary evolution, Theorem \ref{relative} reduces to Theorem \ref{closed} with constant $c=4$. This constant comes from the use of a general bound provided by Acoleyen \cite{Acoleyen}, which is called a small incremental mixing: for two states $\rho_1$ and $\rho_2$, any $0\leq p\leq 1$, and any Hamiltonian $H$ the following bound holds $\Tr(H[p\rho_1, \ln(p\rho_1+(1-p)\rho_2)])\leq -2\Bigl(p\ln p+(1-p)\ln(1-p) \Bigr).$ We used this bound in (\ref{unitary_en}) with $p=1/d$ to derive $\Tr(H [\rho_{AB}, \ln \sigma_0])\leq 4\ln d\|H\|$. If we consider relative entropy of entanglement in the presence of ancillas we would have $1/p=\min\{d_{aA}, d_{Bb}\}$, which depends on the dimension of ancillas. In ancilla assisted case small incremental mixing problem proves small incremental entangling problem with four times bigger constant, i.e. $c=8$. Our "improvement" of the constant came from the fact that we consider ancilla-free system.
The best known constant for the ancilla-free system was given by Bravyi \cite{Bravyi}, which is close to one. We could have used this bound instead of (\ref{unitary_en}), but we chose not to since no precise expression of this constant was given.

\section{Quantum mutual information}

As a measure of correlations between Alice's and Bob's systems one may calculate quantum mutual information between these two systems. The quantum mutual information of a state $\rho_{aABb}$ in a bipartite cut Alice$\--$Bob is as follows:
\begin{align*}
I(aA;Bb)_\rho&=S(\rho_{aA})+S(\rho_{Bb})-S(\rho_{aABb})\\
&=D(\rho_{aABb}|| \rho_{aA}\otimes \rho_{Bb}).
\end{align*}
Since mutual information is differentiable one may investigate the rate of change of quantum mutual information for system evolving according to the irreversible evolution (\ref{irrev}). 
\begin{theorem}
For a system starting in pure state $\rho_{aABb}=\ket{\Psi}\bra{\Psi}_{aABb}$ and evolving according to evolution (\ref{irrev}) the following inequality holds
\begin{align*}
 &\frac{d}{dt}I(aA;Bb)_{\rho(t)}\bigg|_{t=0}\\
 &\leq 4\Bigl(2\|H\|+129\sum_\alpha\|L_\alpha\|^2 \Bigr)(\ln d_A+\ln d_B).
 \end{align*}
\end{theorem}

\begin{proof}
Taking a time derivative of a quantum mutual information we obtain
\begin{align*}
\frac{d}{dt}I(aA;Bb)_{\rho(t)}\bigg|_{t=0}=&-\Tr(\dot{\rho}_{aABb}(t)|_{t=0}\ln\rho_{aA})\\
 &-\Tr(\dot{\rho}_{aABb}(t)|_{t=0}\ln\rho_{Bb})\\
 &+\Tr(\dot{\rho}_{aABb}(t)|_{t=0}\ln\rho_{aABb} ).
\end{align*}
Writing out the dynamics of a state according to (\ref{irrev}), we arrive at
\begin{align*}
& \frac{d}{dt}I(aA;Bb)_{\rho(t)}\bigg|_{t=0}\\
&=\Tr(iH_{AB}[\rho_{aAB}, \ln (\rho_{aA}\otimes\frac{I_B}{d_B})])\\
&-\frac{1}{2}\sum_\alpha\Tr(L_\alpha^\dagger[L_\alpha\rho_{aAB}, \ln(\rho_{aA}\otimes\frac{I_B}{d_B}) ])\\
 &+\frac{1}{2}\sum_\alpha\Tr(L_\alpha[\rho_{aAB} L_\alpha^\dagger, \ln (\rho_{aA}\otimes\frac{I_B}{d_B})])\\
 &+\Tr(iH_{AB}[\rho_{ABb}, \ln (\frac{I_A}{d_A}\otimes\rho_{Bb})])\\
 &-\frac{1}{2}\sum_\alpha\Tr(L_\alpha^\dagger[L_\alpha\rho_{ABb}, \ln(\frac{I_A}{d_A}\otimes\rho_{Bb}) ])\\
 &+\frac{1}{2}\sum_\alpha\Tr(L_\alpha[\rho_{ABb} L_\alpha^\dagger, \ln (\frac{I_A}{d_A}\otimes\rho_{Bb})])\\
 &-\sum_\alpha\Tr(L^\dagger_\alpha[L_\alpha \rho_{aABb}, \ln \rho_{aABb}]).
\end{align*}
According to Lemma 1 in \cite{Bravyi} there exist states $\mu_{aAB}$ and $\sigma_{ABb}$ such that 
$$\rho_{aA}\otimes\frac{I_B}{d_B}=\frac{1}{d_B^2}\rho_{aAB}+(1-\frac{1}{d_B^2})\mu_{aAB}$$
 and
 $$\frac{I_A}{d_A}\otimes\rho_{Bb}=\frac{1}{d_A^2}\rho_{ABb}+(1-\frac{1}{d_A^2})\sigma_{ABb}.$$
Following arguments discussed in previous section leading to (\ref{L-term}) we find that   
$$-i\Tr(H_{AB}[\rho_{aAB}, \ln (\rho_{aA}\otimes\frac{I_B}{d_B})])\leq 8\ln d_{B}\|H\|, $$
$$|\Tr(L_\alpha^\dagger[L_\alpha\rho_{aAB}, \ln\rho_{aA}])|\leq 172\|L_\alpha\|^2\ln(d_B^2)$$
and 
$$-i\Tr(H_{AB}[\rho_{ABb}, \ln (\frac{I_A}{d_A}\otimes\rho_{Bb})])\leq 8\ln d_{A}\|H\|,$$
$$|\Tr(L_\alpha^\dagger[L_\alpha\rho_{ABb}, \ln\rho_{Bb} ])|\leq 172\|L_\alpha\|^2\ln(d_A^2).$$
Therefore the derivative of the mutual information is upper bounded by a quantity independent of dimension of ancillas 
\begin{align*}
& \frac{d}{dt}I(aA;Bb)_{\rho(t)}\bigg|_{t=0}\\
&\leq 4\Bigl(2\|H\|+129\sum_\alpha\|L_\alpha\|^2 \Bigr)(\ln d_A+\ln d_B).
\end{align*}
\end{proof}

\section{Conclusion} 

In this paper we investigated a question of the existence of an upper bound on the entangling rate of an open bipartite system that evolves under an irreversible dynamics.  We considered two cases. First we considered an ancilla-free system and took a relative entropy of entanglement as an entanglement measure. We proved an upper bound on the entangling rate similar to the one in a unitary dynamic case. Having this bound we at the same time have obtained an upper bound on entangling rates for all entanglement measures that are majored by relative entropy of entanglement and reduce to the entanglement entropy on pure states, like \textit{entanglement of distillation} \cite{Bennett96}. 
Second,  we discussed the rate of change of quantum mutual information in ancilla-assisted system and provided an upper bound independent of ancillas.

A question of finding an upper bound on the entanglement rate for any entanglement measure, or a class of entanglement measured, remain open. The small entangling rate problem was applied to the area law problem in the closed system \cite{Acoleyen}, \cite{Marien}. It remains to be seen whether one may find a similar behavior in an open system.

\textbf{Acknowledgments.} A.V. is thankful to David Reeb for helpful suggestions and encouraging conversation. This work was funded through the European Union via QALGO FET-Proactive Project No. 600700.


\begin{thebibliography}{000}

\bibitem{Bravyi} S. Bravyi, Phys. Rev. A \textbf{76}, 052319, (2007).


\bibitem{Acoleyen} K. Van Acoleyen, M. Mari\"{e}n, F. Verstraete,  Phys. Rev. Lett. \textbf{111}, 170501 (2013)

\bibitem{Audenaert} K. M. Audenaert, preprint arXiv:1304.5935.


\bibitem{CLV06} A. M. Childs, D. W. Leung, F. Verstraete,  G. Vidal, Quant. Inf. Comp. \textbf{3}, 97, (2003).

\bibitem{WS03} X. Wang, B.C. Sanders, Phys. Rev. A \textbf{68}, 014301, (2003).

\bibitem{Bennett} C. H. Bennett, A. W. Harrow, D. W. Leung, J. A. Smolin, IEEE Trans. Inf. 
Theory  \textbf{49}, no {8},  1895, (2003).

\bibitem{Lieb} E. H. Lieb, A. Vershynina, Quantum Inf. Comput. \textbf{13}, 0986, (2013).

\bibitem{Marien}M. Mari\"{e}n, K. M.R. Audenaert, K. Van Acoleyen, F. Verstraete, arXiv:1411.0680, (2014).

\bibitem{Raja} A. K. Rajagopal, R. W. Rendell, Phys. Rev. A, 63:022116, (2001).

\bibitem{Zycz} K. Zyczkowski, P. Horodecki, M. Horodecki,  ú
R. Horodecki, Phys. Rev. A, 65:012101, (2001).

\bibitem{Ficek06} Z. Ficek , R. Tanas, Phys. Rev. A, 74:024304, (2006).

\bibitem{Ficek08} Z. Ficek, R. Tanas, Phys. Rev. A, 77(5):054301, (2008).

\bibitem{Lopez} C. E. Lopez, G. Romero, F. Lastra, E. Solano,  J. C. Retamal, Phys. Rev. Lett., 101:080503, (2008).

\bibitem{Bandy} S. Bandyopadhyay, D. A. Lidar, Phys. Rev. A, 70(1):010301, (2004).

\bibitem{Marzolino} U. Marzolino, Europhys. Lett. \textbf{104}, 40004 (2013).

\bibitem{Aolita} L. Aolita, F. de Melo, L. Davidovich, arXiv:1402.3713, (2014).



\bibitem{Lindblad} G. Lindblad, Commun. Math. Phys. \textbf{48}, 119 (1976).

\bibitem{Tiersch} M. Tiersch, F. De Melo, A. Buchleitner,  Journal of Physics A: Mathematical and Theoretical \textbf{46}, 8, 085301 (2013).

\bibitem{Vedral}V. Vedral, M. B. Plenio,Phys. Rev. A, \textbf{57}, 3, 1619 (1998).

\bibitem{Kittaneh} F. Kittaneh,  J.  Func. Anal., \textbf{250}, 132Ð143, (2007).

\bibitem{Bennett96} C. H. Bennett, D. P. DiVincenzo, J. A. Smolin,
W. K. Wootters,  Phys. Rev. A, \textbf{54}, 5, 3824 (1996).

\end{thebibliography}
\end{document}